\documentclass{article}

\usepackage[section]{algorithm}
\usepackage{algorithmic}
\usepackage{amssymb}
\usepackage{amsfonts}
\usepackage{amsmath}
\usepackage{amsthm}
\usepackage{authblk}

\newtheorem{theorem}{Theorem}
\newtheorem{proposition}[theorem]{Proposition}
\newtheorem{corollary}[theorem]{Corollary}
\newtheorem{lemma}[theorem]{Lemma}
\theoremstyle{remark}
\newtheorem{remark}[theorem]{Remark}
\newtheorem{example}[theorem]{Example}
\theoremstyle{definition}
\newtheorem{definition}[theorem]{Definition}
\newtheorem{notation}[theorem]{Notation}

\newcommand{\eqdef}{:=}
\newcommand{\F}{\mathbb{F}}
\newcommand{\C}{\mathcal{C}}
\newcommand{\md}{\delta}

\newcommand{\gen}[1]{\langle #1 \rangle}
\newcommand{\Z}{\mathbb{Z}}
\newcommand{\gggg}{g_1,\ldots,g_r,
                   T^{\ell}(g_1),\ldots,T^{\ell}(g_r),\ldots,
                   T^{(m - 1)\ell}(g_1),\ldots,T^{(m - 1)\ell}(g_r)}
\newcommand{\diam}{\mathbin{\diamond}}
\newcommand{\syn}{\mathcal{S}}
\newcommand{\halfdis}{\lfloor (\delta - 1) / 2 \rfloor}

\DeclareMathOperator\first{\mathcal{F}}
\DeclareMathOperator\pr{pr}

\DeclareMathOperator\ev{ev}
\DeclareMathOperator\GL{GL}
\DeclareMathOperator\qbch{Q-BCH}
\DeclareMathOperator\im{Im}

\DeclareMathOperator\w{w}
\DeclareMathOperator\supp{Supp}

\title{On Quasi-Cyclic Codes as a Generalization of Cyclic Codes}
\author[1]{M. Barbier\thanks{morgan.barbier@lix.polytechnique.fr}} 
\author[2]{C. Chabot\thanks{christophechabotcc@gmail.com}}
\author[1]{G. Quintin\thanks{quintin@lix.polytechnique.fr}}

\affil[1]{\'{E}cole polytechnique,
  Laboratoire d'informatique (LIX),
                91128 Palaiseau Cedex,
                France}
\affil[2]{LJK,
  51 rue des math\'{e}matiques,
  Campus de St Martin d'H\`{e}res BP 53,
  38041 Grenoble Cedex 09,
  France}

\begin{document}
\sloppypar
\maketitle
  \begin{abstract}
    In this article we see quasi-cyclic codes as block cyclic
codes. We generalize some properties of cyclic codes to
quasi-cyclic codes. We show a one-to-one correspondence between
$\ell$-quasi-cyclic codes of length $m\ell$ and left ideals of
$M_{\ell}(\F_q)[X]/(X^m - 1)$. Then, we generalize BCH codes and
evaluation codes in this context. We study their parameters and
establish a key equation. Finally, we present a new
$[189,11,125]_{\F_4}$ code beating the known minimum distance for
fixed length and dimension. Many codes with good parameters beating
best known ones have been found from this latter.

  \end{abstract}

\section{Introduction}
\label{Sec:Intro}

  \subsection{Context}
  Many codes with best known minimum distances are quasi-cyclic codes
or derived from them \cite{LinSol2003,codetables}. This family of
codes is therefore very interesting. Quasi-cyclic codes were 
studied and applied in the context of McEliece's cryptosystems
\cite{McEli78,BerCayGabOtm2009} and Niederreiter's \cite{Nied86,
LiDengWang94}.
They permit to reduce the size of keys in opposition to Goppa
codes. However, since the decoding of random quasi-cyclic codes is
difficult, only quasi-cyclic alternant codes were proposed for the
latter cryptosystems. The high structure of alternant codes is
actually a weakness and two cryptanalyses were proposed in
\cite{FauOtmPerTil2010,UmaLea2010}.
For these reasons, studying the decoding methods and the general
properties of quasi-cyclic codes are interesting topics.

The structure of quasi-cyclic codes has been studied in different
ways. In \cite{LalFitz2001}, quasi-cyclic codes are regarded as
concatenation of cyclic codes, while in \cite{LinSol2001}, the
authors regard them as linear codes over an auxiliary ring. In
\cite{CayChaAbd2010}, the approach is more analogous to the cyclic
case. The authors consider the factorization of
$X^m - 1 \in M_\ell(F_q)[X]$ with reversible polynomials in order to
construct $\ell$-quasi-cyclic codes canceled by those polynomials and
called $\Omega(P)$-codes. This leads to the
construction of self-dual codes and codes beating known bounds.
But the factorization of univariate polynomials over a matrix ring
remains difficult. In \cite{Chabot2011} the author gives an
improved method for particular cases of the latter factorization
problem.

In this article, we prove, analogously to the cyclic case, a
one-to-one correspondence between $\ell$-quasi-cyclic codes of
length $m\ell$ and left ideals of $M_\ell(F_q)[X] / (X^m - 1)$.
We study the properties of quasi-cyclic codes and
propose to extend the definition of BCH and evaluation codes
to the context of quasi-cyclic codes. Namely, we define
\emph{quasi-BCH} and \emph{quasi-evaluation} codes. The natural notion
of \emph{folded} and \emph{unfolded} codes is presented for simplicity
and decoding
purposes. Finally, we exhibit a quasi-cyclic code whose parameters
are better than the previous known and 48 other codes derived from
the first one.

Subsection~\ref{Ssec:Rappel} is devoted to some recalls about
$\Omega(P)$-codes and definitions. Then in
Section~\ref{Sec:Classification} we prove interesting properties
about quasi-cyclic codes and, in particular, the correspondence
between left ideals and quasi-cyclic codes. Section~\ref{Sec:BCH}
deals with the definition, parameters and a decoding algorithm of
quasi-BCH codes. Finally, Section~\ref{Sec:Eval} introduces
quasi-evaluation codes and gives lower bounds on their parameters.

  \subsection{First definitions}
  \label{Ssec:Rappel}
  In this section, we fix a positive integer
$n$ and let $\C$ be a code of length $n$ over the finite field $\F_q$, \emph{i.e.}
a vector subspace of $\F_q^n$.

\begin{definition}[Quasi-cyclic codes]
  From now and until the end of this article
  we define $T:\F_q^n \rightarrow \F_q^n$ to be the left cyclic shift
  defined by:
  \begin{equation*}
    T(c_1, c_2, \dots, c_n) = (c_2,c_3,\dots,c_1).
  \end{equation*}
  Suppose that $\ell$ divides $n$. Then we call
  an \emph{$\ell$-quasi-cyclic code over $\F_q$ of length $n$} a
  code of length $n$ over $\F_q$ stable by $T^\ell$. If the context is clear we will
  simply say \emph{$\ell$-quasi-cyclic code}.
\end{definition}

Let $\ell$ be an integer, and $\alpha \in \F_{q^{\ell}}$ be such that
$(1,\alpha,\ldots,\alpha^{\ell - 1})$ is an $\F_q$-base of the vector
space $\F_{q^{\ell}}$. We define the \emph{folding} to be the $\F_q$-linear map
\begin{equation*}
  \begin{array}{rcl}
    \phi : \F_q^{\ell}    & \rightarrow & \F_{q^{\ell}}
                                          = \F_q[\alpha] \\
    (a_1,\ldots,a_{\ell}) & \mapsto     & a_1 + a_2 \alpha + \cdots +
                                          a_{\ell} \alpha^{\ell - 1}. \\
  \end{array}
\end{equation*}
The unfolding is the inverse $\F_q$-linear map
\begin{equation*}
  \begin{array}{rcl}
    \phi^{-1} : \F_{q^{\ell}} & \rightarrow & \F_q^{\ell} \\
    a = a_1 + a_2 \alpha + \cdots + a_{\ell} \alpha^{\ell - 1}
    & \mapsto & (a_1,a_2,\ldots,a_{\ell}).
  \end{array}
\end{equation*}
Let $m$ be a positive integer, $f : E \rightarrow F$ be any map of
sets. We denote by $f^{\times m}$ the map of sets
$f^{\times m} : E^m \rightarrow F^m$ such that
\mbox{$f^{\times m}(x_1,\ldots,x_m) = (f(x_1),\ldots,f(x_m))$}.

\begin{definition}[Folded and unfolded codes]
  Suppose that $n = m\ell$. We define the \emph{folded code} of $\C$
  to be $\phi^{\times m}(\C)$. Let $\C'$ be a code in
  $\F_{q^{\ell}}^m$. We define the \emph{unfolded code} of
  $\C'$ to be $(\phi^{-1})^{\times m}(\C')$.
\end{definition}

\begin{remark}
  Observe that a code $\C$ is $\ell$-quasi cyclic if and only if
  its folded $\C' = \phi^{\times m}(\C)$ is cyclic. But $\C'$
  is not necessarily $\F_{q^{\ell}}$-linear.
\end{remark}

\section{Properties of quasi-cyclic codes}
\label{Sec:Classification}
In the present section we generalize the results of
\cite[Theorem~1, page~190]{SloMacWil86} to quasi-cyclic codes.
We fix a positive integer $n$ and suppose that
$n = m\ell$ for two positive integers $m$ and $\ell$.

\subsection{The one-to-one correspondence}

It is well-known \cite[Theorem~1, page~190]{SloMacWil86} that there
is a one-to-one correspondence between cyclic codes of length $n$
over $\F_q$ and monic factors of $X^n - 1 \in \F_q[X]$ \emph{i.e.}
ideals of $\F_q[X] / (X^n - 1)$.
In \cite{CayChaAbd2010,Chabot2011} the authors start to exhibit
such a correspondence for quasi-cyclic codes. They show that there
is a correspondence between a subfamily of $\ell$-quasi-cyclic codes
of length $m\ell$ over $\F_q$ and reversible factors of
$X^n - 1 \in M_{\ell}(\F_q)[X]$.

The one-to-one correspondence between $\ell$-quasi cyclic codes
and left ideals of $M_{\ell}(\F_q)[X] / (X^m - 1)$ is a consequence
of the two following lemmas.

\begin{lemma}
\label{lem:mod_gen}
  Let $R$ be a commutative principal ring and $M$ be a free left
  module of finite rank $s$ over $R$. Then every submodule $N$ of $M$
  can be generated by at most $s$ elements.
\end{lemma}

\begin{proof}
  It is an easy adaptation of the proof of
  \cite[Theorem~7.1, page~146]{Lang2002}.
\end{proof}

\begin{lemma}
\label{lem:morita}
  Let $s$ be a positive integer and $R$ be a commutative principal
  ring. Then there is a one-to-one correspondence between the
  submodules of $R^s$ and the left ideals of $M_s(R)$.
\end{lemma}

\begin{proof}
  Note that this is a particular case of the Morita equivalence
  for modules. See for example
  \cite[n\textsuperscript{o}4, page~99]{BourbakiAlgCom8}. This
  particular case can be proved directly.
  To a submodule $N \subseteq R^s$, we can build a left ideal
  of $M_s(R)$ whose elements have rows in $N$. Conversely,
  to a left ideal $I \subseteq M_s(R)$ we associate the
  submodule of $R^s$ generated by all the rows of all the elements
  of $I$. It is straightforward to check that these maps are
  inverse to each other.
\end{proof}

Note that $M_{\ell}(\F_q)[X] / (X^m - 1)$ and
$M_{\ell}(\F_q[X] / (X^m - 1))$ are isomorphic as rings and that
$R = \F_q[X] / (X^m - 1)$ is a commutative principal ring. By
Lemma~\ref{lem:mod_gen} any submodule of $R^{\ell}$ can be generated
by at most $\ell$ elements. Therefore by Lemma~\ref{lem:morita}
any left ideal of $M_{\ell}(R) = M_{\ell}(\F_q)[X] / (X^m - 1)$ is
principal.

\begin{theorem}
\label{thm:one-to-one}
  There is a one-to-one correspondence between $\ell$-quasi-cyclic
  codes over $\F_q$ of length $m\ell$
  and left ideals of $M_{\ell}(\F_q)[X] / (X^m - 1)$.
\end{theorem}

\begin{proof}
  Let $g = (g_{11},\ldots,g_{1\ell},g_{21},\ldots,g_{2\ell},\ldots,
       g_{m1},\ldots,g_{m\ell}) \in \F_q^{m\ell}$. We associate
  to $g$ the element $\varphi(g) \in (\F_q[X] / (X^m - 1))^{\ell}$
  defined by
  \begin{multline*}
    \varphi(g) =
    \left(
      g_{11} + g_{21} X + \cdots + g_{m1} X^{m - 1} ;
    \right. \\
    g_{12} + g_{22} X + \cdots + g_{m2} X^{m - 1} ; \ldots ; \\
    \left.
      g_{1\ell} + g_{2\ell} X + \cdots + g_{m\ell} X^{m - 1}
    \right).
  \end{multline*}
  Then $\varphi$ induces a one-to-one correspondence between
  $\ell$-quasi-cyclic codes of length $m\ell$ over $\F_q$
  and submodules of $(\F_q[X] / (X^m - 1))^{\ell}$.
  The theorem follows by Lemma~\ref{lem:morita}.
\end{proof}

Let $\pr_{i,j}$ be the projection of the $i,i + 1,\ldots,j$ coordinates:
\begin{equation*}
  \begin{array}{rcl}
    \pr_{i,j} : \F_q^n & \longrightarrow & \F_q^{j - i + 1} \\
    (x_1,\ldots,x_n)   & \longmapsto     &
                      (x_i,x_{i + 1},\ldots,x_{j - 1},x_j).
  \end{array}
\end{equation*}
We have the following obvious lemma:

\begin{lemma}
\label{LemmaBlockRank}
  Let $\C$ be an $\ell$-quasi-cyclic code over $\F_q$ of dimension $k$
  and length $m\ell$.
  Then there exists an integer $r$ such that $1 \leq r \leq k$ and
  for any generator matrix $G$ of $\C$ and $0 \leq i \leq m - 1$,
  the rank of the $i\ell + 1,i\ell + 2,\ldots,(i + 1)\ell$ columns of
  $G$ is $r$.
\end{lemma}

\begin{definition}[Block rank]
  Taking the notation of Lemma~\ref{LemmaBlockRank},
  we call the integer $r$ the \emph{block rank} of $\C$. Note that
  $r$ depends only on $\C$ and not on any particular generator matrix
  of $\C$.
\end{definition}

\subsection{The generator polynomial of an $\ell$-quasi-cyclic code}

In this subsection we fix an $\ell$-quasi-cyclic code $\C$ over $\F_q$.
If $\ell = 1$, then $\C$ is a cyclic code of length $n$ and a
generator matrix of $\C$ can be given
\cite[Theorem~1,~(e), page~191]{SloMacWil86} by
\begin{equation}
\label{equ:cyclic-gen}
  \begin{pmatrix}
    g(X)  &       &       & \\
          & Xg(X) &       & \\
          &       & \dots & \\
          &       &       & X^{n - \deg g}g(X)
  \end{pmatrix},
\end{equation}
where $g(X) \in \F_q[X]$ is the generator polynomial of $\C$.
The block rank of $\C$ is $1$ and we see that we can write
a generator matrix of $\C$ with only $1$ vector and its shifts
(by $T^{\ell} = T$). The natural generalization of this result
for quasi-cyclic codes is done using the block rank.

Let $r$ be the block rank of $\C$, the following algorithm computes
a basis of $\C$ from $r$ vectors of $\C$ and their shifts.
We call the \emph{first index} of a nonzero vector
$x = (x_1,\ldots,x_{m\ell})$ the least integer $0 \leq i \leq m - 1$
such that \mbox{$(x_{i\ell + 1},\ldots,x_{(i + 1)\ell}) \neq 0$} and denote
it by $\first(x) = \first(x_1,\ldots,x_{m\ell})$.
Let
\begin{equation*}
  \begin{array}{rlc}
  p:\F_q^{m\ell} &\longrightarrow & \F_q^{\ell}\\
  x = (x_1,\ldots,x_{m\ell}) & \longmapsto &
        (x_{i\ell + 1},\ldots,x_{(i + 1)\ell}),
  \end{array}
\end{equation*}
where $i = \first(x_1,\ldots,x_n)$ if $x \neq 0$ and $p(0) = 0$.

\begin{algorithm}
\label{al:basis-block-rank}
\caption{Basis computation with the block rank}
\begin{algorithmic}[1]
\REQUIRE A generator matrix $G$ of $\C$.
\ENSURE A generator matrix formed by $r$ rows from $G$
        and some of their shifts.
\STATE $G' \gets $ a row echelon form of $G$.
\STATE Denote by $g_1,\ldots,g_k$ the rows of $G'$.
\STATE $M \gets \max \{ \first(g_i) : i \in \{ 0,\ldots,m - 1 \} \}$.
\STATE $B_M' \gets \emptyset$.
\STATE $G_{M + 1} \gets \emptyset$.
\FOR{$j = M \to 0$}
  \STATE $B_j \gets$ $\{ g_i : i \in \{ 1,\ldots,k \}
         \text{ and } \first(g_i) = j \}$.
  \FOR{each element $x$ of $B_j$}
    \IF{$p(B_j') \cup \{ p(x) \}$ are independent}
      \STATE $B_j' \gets B_j' \cup \{ x \}$.
    \ENDIF
  \ENDFOR
  \STATE $G_j \gets G_{j + 1} \cup B_j'$.
  \STATE $B_{j - 1}' \gets T^{\ell}(B_j')$.
\ENDFOR
\RETURN $G_0$.
\end{algorithmic}
\end{algorithm}

Note that Algorithm~\ref{al:basis-block-rank} applied to a cyclic
code, \emph{i.e.} $\ell = 1$, returns exactly the
matrix~\eqref{equ:cyclic-gen} and we can deduce the generator
polynomial of $\C$ at the cost of the computation of a row echelon
form of any generator matrix of $\C$.

\begin{proposition}  
  Algorithm~\ref{al:basis-block-rank} works correctly as expected and
  returns a generator matrix $G$ of $\C$ made of $r$ linearly
  independent vectors of $\C$ and some of their shifts.
\end{proposition}

\begin{proof}
  We will prove by descending induction on $j$ that:
  \begin{enumerate}
    \item $B_j' \supseteq T^{\ell} (B_{j + 1}')
                \supseteq \dots
                \supseteq T^{(M - j)\ell} (B_M')$.
    \item $\#B_j' \leq r$.
    \item The vectors of $B_j'$ are linearly independent.
    \item The vectors of $G_j$ are linearly independent.
    \item $\gen{G_j} = \gen{g_i : i \in \{ 1,\ldots,k \}
                                  \text{ and } \first(g_i) \geq j }$.
  \end{enumerate}
  Let $j = M$. By step~3, we have $B_M \neq \emptyset$. Item~1
  is trivially satisfied. By Lemma~\ref{LemmaBlockRank},
  $\#B_M \leq r$ and item~2 is satisfied. As
  $G_{M + 1} = B_M' = \emptyset$ then
  $G_M = B_M' = B_M = \{ g_i : i \in \{ 1,\ldots,k \}
                               \text{ and } \first(g_i) \geq M \}$ and
  items~3 to~5 are satisfied.

  Suppose that $j < M$ and that items~1 to~5 are satisfied for
  $i = j + 1,\ldots,M$. First note that $B_j \neq \emptyset$.
  If we had $B_j = \emptyset$ then, as $G'$ is in row echelon form,
  $g_1,\ldots,g_k,T^{(M - j)\ell}(g_k)$ would be linearly independent
  which is a contradiction.

  Items~1 and~3 are satisfied by steps~7, 9 and~10 of the algorithm. By
  Lemma~\ref{LemmaBlockRank} and step~9, item~2 is satisfied. For all
  $x \in G_{j + 1}$, we have $\first(x) \geq j + 1$, thus, by
  item~3, the elements of $G_j$ are linearly independent and item~4
  is satisfied. Let $g$ be a vector of $G'$ such that $\first(g) = j$, then
  the construction of $B_j'$ implies that we have
  \begin{equation*}
    \first \left( g - \sum_{u \in B_j'} \mu_u u \right) \geq j + 1
  \end{equation*}
  where $\mu_u \in \F_q$ for $u \in B_j'$. Then by item~5 of the
  inductive hypothesis, we have
  \begin{equation*}
    \left( g - \sum \mu_u u \right) \in G_{j + 1}.
  \end{equation*}
  Thus we have
  $\gen{G_j} = \gen{ g_i : i \in \{1,\ldots,k\}
                           \text{ and } \first(g_i) \geq j }$ and item~5
  is satisfied.

  As a consequence of the previous induction, $G_0$ is constituted of
  linearly independent vectors and generates
  $\gen{ g_i : i \in \{1,\ldots,k\} \text{ and } \first(g_i) \geq 0 }
   = \C$ by item~5. By Lemma~\ref{LemmaBlockRank} we must have
  exactly $r$ vectors $g \in G_0$ such that $\first(g) = 0$. Thus
  by items~1 and~2 we have
  \begin{equation*}
    r = \#B_0' = \sum_{\lambda = 0}^M
                 \#\left(
                   B_{\lambda}' \setminus T^{\ell} (B_{\lambda + 1}')
                 \right)
  \end{equation*}
  which shows that $G_0$ is constituted of $r$ linearly independent
  vectors of $\C$ and some of their shifts.
\end{proof}

\begin{corollary}
\label{cor:gen}
  There exist $g_1,\ldots,g_r$ linearly independent vectors of $\C$
  such that $\gggg$ span $\C$.
  If we denote by $g_{i,j}$ the $j$'th coordinate of $g_i$ and let
  \begin{equation*}
    G_i =
    \begin{pmatrix}
      g_{1,i\ell + 1} & \dots & g_{1,(i + 1)\ell} \\
      \vdots          &       & \vdots \\
      g_{r,i\ell + 1} & \dots & g_{r,(i + 1)\ell} \\
                      & 0     &
    \end{pmatrix}
    \in M_{\ell}(\F_q)
  \end{equation*}
  and
  \begin{equation*}
    g(X) = \frac{1}{X^{\nu}} \sum_{i = 0}^{m - 1} G_i X^i
           \in M_{\ell}(\F_q)[X],
  \end{equation*}
  where $\nu$ is the least integer such that $G_i \neq 0$,
  then $\C$ corresponds to the left ideal $\gen{g(X)}$
  by Theorem~\ref{thm:one-to-one}.
\end{corollary}

\begin{corollary}
  Taking the notation of the proof of Theorem~\ref{thm:one-to-one},
  the submodule $\varphi(\C) \subseteq (\F_q[X] / (X^m - 1))^{\ell}$
  is generated by $r$ elements as an $\F_q[X] / (X^m - 1)$-module
  but cannot be generated by less that $r$ elements. If $\C$ is a
  cyclic code then we have $r = 1$ and we find the classical result
  about cyclic codes.
\end{corollary}

\begin{definition}[Generator polynomial]
  The polynomial $g(X) \in M_{\ell}(\F_q)[X]$ from
  Corollary~\ref{cor:gen} is called a \emph{generator polynomial}
  of $\C$.
\end{definition}

\begin{example}
  Let $I = \gen{P(X),Q(X)} \subset M_3(\F_4)[X]/(X^5-1)$ be a left ideal.
  The row echelon form generator matrix of the $3$-quasi cyclic code
  $\C_I$ associated to the left ideal $I$ is 
  \begin{equation*}
    G = \left( \begin{array}{ccc|ccc|ccc|ccc|ccc}
      1 & 0 & \omega^2 &
      0 & 0 & 0 &
      0 & \omega^2 & \omega &
      \omega & 0 & 1 &
      0 & 0 & 0 \\
 
      0 & 1 & \omega^2 &
      0 & 0 & 0 &
      0 & 0 & 0 &
      \omega & \omega & 0 &
      1 & 0 & \omega^2 \\
      
      \hline
      
      0 & 0 & 0 &
      1 & 0 & \omega^2 &
      0 & 0 & 0 &
      0 & \omega^2 & \omega &
      \omega & 0 & 1 \\
 
      0 & 0 & 0 &
      0 & 1 & \omega^2 &
      0 & \omega^2 & \omega &
      \omega & 0 & 1 &
      \omega & \omega & 0 \\
 
      \hline
 
      0 & 0 & 0 &
      0 & 0 & 0 &
      1 & 1 & 0 &
      \omega^2 & 0 & \omega &
      0 & \omega^2 & \omega 
    \end{array} \right).
  \end{equation*}
  Algorithm~\ref{al:basis-block-rank} gives that
  $(g_4,g_5,T^3(g_4),T^3(g_5),T^{2 \times 3}(g_5))$ is a basis of
  $\C_I$. Moreover
  \begin{equation*}
    g(X) =
    \begin{pmatrix}
    0 & 1 & \omega^2 \\
    0 & 0 & 0 \\
    0 & 0 & 0
    \end{pmatrix} +
    \begin{pmatrix}
    0 & \omega^2 & \omega \\
    1 & 1 & 0 \\
    0 & 0 & 0
    \end{pmatrix} X +
    \begin{pmatrix}
    \omega & 0 & 1 \\
    \omega & 0 & \omega \\
    0 & 0 & 0
    \end{pmatrix} X^2 +
    \begin{pmatrix}
    \omega & \omega & 0 \\
    0 & \omega^2 & \omega \\
    0 & 0 & 0
    \end{pmatrix} X^3
  \end{equation*}
  is a generator polynomial of $\C_I$ and
  $I=\gen{P(X),Q(X)} = \gen{g(X)}$.
\end{example}

\subsection{A property of generator polynomials}

The following proposition generalizes
\cite[Theorem~1,~(c), page~190]{SloMacWil86} and
\cite[Theorem~4, page~196]{SloMacWil86}.

\begin{proposition}
  Let $\C$ be an $\ell$-quasi-cyclic code of length $m\ell$ over $\F_q$.
  Let $P(X)$ be a generator polynomial of $\C$ and $Q(X)$ a
  generator polynomial of its dual.
  Then 
  \begin{equation*}
    P(X) \left( ^t Q^{\star}(X) \right) = 0 \pmod{X^m - 1}
  \end{equation*}
  where $Q^{\star}$ denotes the reciprocal polynomial of $Q$ and $^t Q$
  the polynomial whose coefficients are the transposed matrices of
  the coefficients of $Q$.
\end{proposition}

\begin{proof}
  Since $P(X) = \sum_{i = 0}^{m - 1} P_i X^i$ is a generator
  polynomial of $\C$, the rows of the matrix
  \begin{equation*}
    \begin{pmatrix} P_0 & P_1 & \ldots & P_{m - 1} \end{pmatrix}
  \end{equation*}
  and their shifts span $\C$.
  Similarly $Q(X) = \sum_{i = 0}^{m - 1} Q_i X^i$ and the rows of
  \begin{equation*}
    \begin{pmatrix} Q_0 & Q_1 & \ldots & Q_{m - 1} \end{pmatrix}
  \end{equation*}
  and their shifts span $\C^{\perp}$.
  By definition of a dual code, we have
  \begin{equation*}
    \begin{pmatrix} P_0 & P_1 & \cdots & P_{m-1} \end{pmatrix}
    \begin{pmatrix} ^t Q_0 \\ ^t Q_1 \\ \vdots \\ ^t Q_{m-1} \end{pmatrix}
    = \sum_{i=0}^{m-1} P_i \left( ^t Q_i \right) = 0.
  \end{equation*}
  As $\C$ and $\C^{\perp}$ are $\ell$-quasi cyclic codes we also have
  \begin{equation*}
    \sum_{i = 0}^{m - 1} P_i \left( ^t Q_{i + j \mod m} \right) = 0
  \end{equation*}
  for all $j \in \Z$. Therefore 
  \begin{equation*}
    P(X) \left( ^t Q^{\star}(X) \right) =
    \sum_{j = 0}^{m-1} \sum_{i = 0}^{m - 1}
    P_i \left( ^tQ_{i - j \mod m} \right) X^j
    = 0 \mod (X^m-1).
  \end{equation*}
  Hence the proposition.
\end{proof}

\section{Quasi-BCH}
\label{Sec:BCH}
In Section~\ref{Sec:Classification} we saw that
quasi-cyclic codes can be regarded as a generalization of
cyclic codes. Therefore, it is interesting to focus on the
generalization of BCH codes. We start with the definition and then
study their parameters. Finally we present a decoding scheme for
quasi-BCH codes raising interesting questions. We fix four positive
integers $n = m\ell$ and $s$.

\subsection{Definition}

\begin{definition}[Primitive root of unity]
\label{defi:racprim}
  Let $q$ be a prime power. A matrix
  $A \in M_\ell(\F_{q^s})$ is called a \emph{primitive $m$-th root of
  unity} if
  \begin{itemize}
    \item $A^m = I_\ell$,
    \item $A^i \neq I_\ell$ if $i < m$,
    \item $\det(A^i - A^j) \neq 0$, whenever $i \neq j$.
  \end{itemize}
\end{definition}

\begin{proposition}
\label{prop:racprim}
  Let $q$ be a prime power and suppose that $q^{s\ell} - 1 = m$.
  Then there exists a primitive $m$-th root of unity in
  $M_{\ell}(\F_{q^s})$.
\end{proposition}

\begin{proof}
  Let $\alpha \in \F_{q^{s\ell}}$ be a primitive $m$-th root of unity
  and $A \in M_{\ell}(\F_{q^s})$ be the companion matrix of the
  irreducible polynomial $f(X) \in \F_{q^s}[X]$ of $\alpha$ over
  $\F_{q^s}$. There exists $P \in \GL_{\ell}(\F_{q^{s\ell}})$ and an
  upper triangular matrix $U \in M_{\ell}(\F_{q^{s\ell}})$ whose
  diagonal coefficients are the eigenvalues of $A$ such that
  $A = P^{-1} U P$. The eigenvalues of $A$ are exactly the
  roots of $f$ and then are primitive $m$-th roots of unity. Therefore
  $A$ satisfies the three conditions of Definition~\ref{defi:racprim}.
\end{proof}

\begin{definition}[Block minimum distance]
  Let $\C$ be a linear code over $\F_q$ of length $m\ell$.
  We define the \emph{$\ell$-block minimum distance} of $\C$ to be
  the minimum distance of the folded code of $\C$.
\end{definition}

\begin{definition}[Left quasi-BCH codes]
\label{defi:QBCH}
  Let $A$ be a primitive $m$-th root of unity in
  $M_{\ell}(\F_{q^s})$ and $\delta \leq m$.
  We define the $\ell$-quasi-BCH code of length $m\ell$,
  with respect to $A$, with designed minimum distance $\delta$,
  over $\F_q$ by
  \begin{multline*}
    \qbch_q(m,\ell,\delta,A) \eqdef \\
    \left\lbrace
      (c_1,\ldots,c_m) \in (\F_q^\ell)^m :
        \sum_{j = 0}^{m - 1} A^{ij}c_j = 0
        \text{ for } i = 1,\ldots,\delta - 1
    \right\rbrace.
  \end{multline*}
  We call the linear map
  \begin{equation*}
    \begin{array}{rcl}
      \syn_A : (\F_q^{\ell})^m & \rightarrow & (\F_{q^s}^{\ell})^m \\
      x = (x_1,\ldots,x_m)     & \mapsto     & \sum_{j = 0}^{m - 1}
                                                     A^j x_j
    \end{array}
  \end{equation*}
  the \emph{syndrome} map with respect to $\qbch(m,\ell,\delta,A)$.
\end{definition}

\begin{proposition}
  Using the notation of Definition~\ref{defi:QBCH},
  $\qbch_q(m,\ell,\delta,A)$ has dimension at least
  $(m - e(\md - 1))\ell$ and $\ell$-block minimum distance at least
  $\delta$. In other words $\qbch_q(m,\ell,\delta,A)$ is an
  $[m\ell , \geq (m - s(\delta - 1))\ell , \geq \delta]_{\F_q}$-code.
\end{proposition}

\begin{proof}
  According to Definition~\ref{defi:QBCH} we have that
  \begin{equation*}
    H = \begin{pmatrix}
      I_{\ell} & A              & \cdots & A^{m - 1} \\
      I_{\ell} & A^2            & \cdots & A^{2(m - 1)} \\
      \vdots   & \vdots         &        & \vdots \\
      I_{\ell} & A^{\delta - 1} & \cdots & A^{(\delta - 1)(m - 1)}
    \end{pmatrix} \in M_{(\delta - 1)\ell,m\ell}(\F_{q^s})
  \end{equation*}
  is a parity check matrix of $\qbch_q(m,\ell,\delta,A)$.
  Let
  \begin{equation*}
    V = \begin{pmatrix}
      I_\ell & A              & \cdots & A^{\delta - 1} \\
      I_\ell & A^2            & \cdots & A^{2(m - 1)} \\
      \vdots & \vdots         &        & \vdots \\
      I_\ell & A^{\delta - 1} & \cdots & A^{(\delta - 1)^2}
    \end{pmatrix}.
  \end{equation*}
  Using the Vandermonde matrix trick we find that the determinant $D$
  of $V$ over $M_{\ell}(\F_{q^s})[A]$ is $\prod_{i < j} (A^i - A^j)$.
  By the definition of $A$ we have $\det_{\F_{q^s}} D \neq 0$, thus
  $V$ is invertible over $M_{\ell}(\F_{q^s})[A]$ and then, invertible
  over $\F_{q^s}$. Therefore $H$ has full rank over $\F_{q^s}$.
  
  Let $i:\F_q^{m\ell} \rightarrow \F_{q^s}^{m\ell}$ be the canonical
  injection and denote by
  $h:\F_{q^s}^{m\ell} \rightarrow \F_{q^s}^{(\delta - 1)\ell}$
  the $\F_q$-linear map given by $H$. Then we have
  $\dim_{\F_q}(\im h) = e(\delta - 1)\ell$. Thus
  $\dim_{\F_{q^s}}(\im h \circ i) \leq (\delta - 1)\ell$ and
  $\dim_{\F_q}(\im h \circ i) \leq e(\delta - 1)\ell$. Therefore
  $\dim_{\F_q}(\ker h \circ i) \geq m\ell - e(\delta - 1)\ell$.
  Suppose that there exists a codeword
  \mbox{$c = (c_1,\ldots,c_m) \in \C \setminus \{0\}$}
  with $\ell$-block weight $b \leq \delta - 1$. Note $i_1,\ldots,i_b$
  the indexes such that $c_{i_j}\neq 0$ for $i=1,\ldots,b$.
  This implies that the matrix
  \begin{equation*}
    \begin{pmatrix}
      A^{i_1}  & A^{i_2}  & \cdots & A^{i_b} \\
      A^{2i_1} & A^{2i_2} & \cdots & A^{2i_b} \\
      \vdots   & \vdots   &        & \vdots \\
      A^{(\delta - 1) i_1} &
        A^{(\delta - 1) i_2} &
          \cdots &
            A^{(\delta - 1) i_b} \\
    \end{pmatrix}
  \end{equation*}
  has not full rank which is absurd.   
\end{proof}

\begin{example}
  Consider the $3$-quasi-BCH codes defined by primitive roots in
  $M_3(\F_{2^2})$ of length $63$ over $\F_2$ with designed
  minimum distance $6$ defined by a $21$-th root of unity in
  $\F_{2^2}$. In other words, $q = 2,m = 21,\ell = 3,s = 2$ and
  $\delta = 6$. There are $22$ non-equivalent codes splitting as
  follows:
  \begin{equation*}
    \begin{array}{|c|c|}
      \hline
        \text{Number of codes} & \text{Parameters} \\
      \hline
        2  & [63,33,6]_{\F_2} \\
      \hline
        18 & [63,33,7]_{\F_2} \\
      \hline
        2  & [63,36,6]_{\F_2} \\
      \hline
    \end{array}
  \end{equation*}
  Notice that their dimension is always at least
  $(m - e(\delta - 1))\ell = 33$ and their minimum distance is
  at least $\delta = 6$. All the computations have been performed
  with the \textsc{magma} computer algebra system \cite{magma}.
\end{example}

\begin{example}
  Let $q = 5,m = 7,\ell = 3,s = 2$ and $\delta = 3$. Let
  $\omega \in \F_{5^2}$ be a primitive $(5^2 - 1)$-th root of unity and
  \begin{equation*}
    A=\begin{pmatrix}
      \omega^9    & \omega^4    & \omega^{22} \\
      \omega^{11} & \omega^{11} & \omega^{15} \\
      \omega^2    & \omega^{19} & 1 
    \end{pmatrix}
    \in M_3(\F_{5^2}).
  \end{equation*}
  Then the left $3$-quasi-BCH code of length $21$ with respect
  to $A$ with designed minimum distance $3$ over $\F_5$ has parameters
  $[21,9,7]_{\F_5}$. Its generator polynomial is given by
  \begin{multline*}
    g(X) = 
    \begin{pmatrix} 1 & 4 & 3 \\ 3 & 3 & 4 \\ 1 & 1 & 4 \end{pmatrix} X^4 +
    \begin{pmatrix} 4 & 0 & 0 \\ 4 & 0 & 0 \\ 4 & 0 & 4 \end{pmatrix} X^3 +
    \begin{pmatrix} 3 & 0 & 4 \\ 0 & 3 & 4 \\ 0 & 0 & 0 \end{pmatrix} X^2 +\\
    \begin{pmatrix} 2 & 3 & 2 \\ 4 & 4 & 4 \\ 3 & 1 & 1 \end{pmatrix} X +
    \begin{pmatrix} 1 & 0 & 0 \\ 0 & 1 & 0 \\ 0 & 0 & 1 \end{pmatrix}
  \in M_3(\F_5)[X].
  \end{multline*}
\end{example}

\section{Decoding scheme for quasi-BCH codes}
\label{Sec:Key}
For this section we fix five positive integers $n = m\ell$, $r$ and
$\delta$, a primitive $m$-th root of unity $A \in M_{\ell}(\F_{q^s})$
and $\C = \qbch(m,\ell,\delta,A)$.
If the folded of $\C$ is a BCH code $\C'$ over $\F_{q^{\ell}}$ (which
is not the case in general) then
we can apply the standard, unique and list, decoding algorithms.
See for example \cite[Paragraph~6, page~270]{SloMacWil86} and
\cite{AugBarCou2011}. If $\C'$ is not a code for which a decoding
algorithm is known, we propose in what follows a decoding scheme for
$\C$ based on the key equation that we establish for quasi-BCH codes.
Following the same techniques as for BCH codes, we first compute the
locator and evaluator polynomials by solving the key equation and
then compute the error vector and recover the original message.

\begin{notation}
\label{not:utils}
  Let $\kappa$ be any field and $x = (x_1,\ldots,x_n) \in \kappa^n$.
  We denote by $w(x)$ the Hamming weight of $x$ \emph{i.e.} the
  cardinal of
  \mbox{$W = \{ i : i \in \{ 1,\ldots,n \} \text{ s.t. } x_i \neq 0 \}$}.
  We denote by $\supp(x)$ the support of $x$ \emph{i.e.} the set $W$.
\end{notation}

\subsection{The key equation}

As in the scalar case, we exhibit a key equation for quasi-BCH codes.
In this subsection, all vectors are considered to be single-column
matrices. Consider $\F_q^{\ell}$ as a product ring of $\ell$ copies of
$\F_q$. We define a map
\begin{equation*}
  \begin{array}{rcl}
    \Psi : M_{\ell}(\F_{q^s})[[X]] \times \F_q^{\ell}[[X]] &
              \rightarrow & \F_{q^s}^{\ell}[[X]] \\
    (f,g) & \mapsto & \sum_{i,j} f_j g_i X^{i + j}
  \end{array}
\end{equation*}
where the $f_i g_j$ are matrix-vector products. In the sequel we will
denote $\Psi(f,g)$ simply by $f \diam g$. Note that we have
$(fh) \diam g = f \diam (h \diam g)$ for any $h \in M_{\ell}(\F_{q^s})$.

Let $c$ be a codeword of $\C$ sent over a channel,
$y \in (\F_q^{\ell})^m$ be the received word and let $e$ be
the error vector \emph{i.e.} $e = y - c$ such that
$\w(e) = w \leq \halfdis$. Let
$W = \supp(e) = \{ i_1,\ldots,i_w \}$.
%%$A_j = A^{i_j}$ and
%%$y_j = e_{i_j} \in \F_q^{\ell}$ for $j = 1,\ldots,w$.

\begin{definition}[Locator and evaluator polynomials]
  We define the \emph{locator polynomial} by
  \begin{equation*}
    \Lambda(X) \eqdef \prod_{i \in W} (1- A^i X)
    \in M_{\ell}(\F_{q^s})
  \end{equation*}
  and the \emph{evaluator polynomial} by
  \begin{equation*}
    L(X) \eqdef \sum_{i \in W}
    \left(
      \prod_{j \neq i}^{w}
      A^i (1 - A^j) X
    \right) \diamond y_i
    \in \F_{q^s}^{\ell}[X].
  \end{equation*}
\end{definition}

\begin{lemma}
  \label{lem:invertSerie}
  Let $B \in M_\ell(\F_q)$ be a nonzero matrix, then $1 - BX$ has a
  left- and right- inverse in $M_\ell(\F_q)[[X]]$, both equal to
  \begin{equation*}
    \sum_{j = 0}^{+\infty} B^j X^j.
  \end{equation*}
\end{lemma}

We see that the locator polynomial $\Lambda(X)$ is invertible in
the power series ring $M_{\ell}(\F_{q^s})[[X]]$ and we have

\begin{align*}
  \left( \Lambda(X)^{-1} \right) \diamond L(X) &=
      \sum_{i \in W} \left(
        A^i (1 - A^i X)^{-1}
      \right) \diam y_i \\
  &= \sum_{i \in W} \left(
       \sum_{j = 0}^{+\infty} A^{i(j + 1)} X^j
     \right) \diam y_i \\
  &= \sum_{j = 0}^{+\infty}
       \sum_{i \in W} A^{i(j + 1)} y_i X^j.
\end{align*}

Using the fact that $y = c + e$ and that, by definition,
$\syn_{A^i}(y) = \syn_{A^i}(e)$ for any $i = 0,\ldots,\delta - 1$
we have

\begin{equation*}
  \left( \Lambda(X)^{-1} \right) \diamond L(X) =
     \sum_{j = 0}^{+\infty} \syn_{A^{j + 1}}(e) X^j
       \eqdef S_{\infty}(X).
\end{equation*}

\begin{proposition}
\label{prop:KE}
  For any error vector $e \in \F_q^{m\ell}$ such that
  $w(e) \leq \halfdis$ we have
  \begin{center}
    \fbox{$\Lambda(X) \diam S_{\infty}(X) = L(X)$}
  \end{center}
  and therefore
  \begin{equation}
  \label{equ:KE}
    \Lambda(X) \diam S_{\infty}(X) \equiv L(X) \mod X^{\delta}.
  \end{equation}
  We will refer to~\eqref{equ:KE} as the \emph{key equation}.
\end{proposition}

\subsubsection{Problems solving the key equation}
\label{sss:solving_ke}

In the case of BCH codes, the extended Euclidean and Berlekamp-Massey
algorithms can be used to solve the key equation.
We denote by $S_{\delta}(X)$ the polynomial
$S_{\infty}(X) \mod X^{\delta}$ from~\eqref{equ:KE} which can
be written as

\begin{equation}
\label{equ:lambda_L}
  \begin{pmatrix}
    \Lambda_0 & \dots & \Lambda_{\delta - 1} &
        \vline & L_0 & \dots & L_{\delta - 1}
  \end{pmatrix}
  \begin{pmatrix}
    S_0 & S_1    & \dots  & S_{\delta - 1} \\
        & S_0    &        & \vdots \\
        &        & \ddots & \vdots \\
        &        &        & S_0 \\
    \hline
    -1     & 0      & \dots  & 0 \\
     0     & -1     &        & \vdots \\
    \vdots &        & \ddots & 0 \\
     0     & \dots  & 0      & -1
  \end{pmatrix}
  = 0.
\end{equation}

Where the $S_i$'s and $L_i$'s are column vectors such that
the $S_i$'s are the coefficients of $S_{\delta}$ in
$\F_{q^s}^{\ell}$ and the $L_i$'s are
the coefficients in $\F_{q^s}^{\ell}$ of $L(X)$.
The $\Lambda_i$'s are the coefficients of $\Lambda(X)$ in
$M_{\ell}(\F_{q^s})$.
This system of linear equations over $\F_{q^s}$ has many solutions in
$\F_{q^s}$ since there are $\ell\delta + \delta$ unknowns and only
$\delta$ equations for each row of
\begin{equation*}
  \begin{pmatrix}
    \Lambda_0 & \dots & \Lambda_{\delta - 1} &
        \vline & L_0 & \dots & L_{\delta - 1}
  \end{pmatrix}.
\end{equation*}
However, we are only interested in the solution such that
$(\Lambda_0,\ldots,\Lambda_{\delta - 1})$ is an error locator
polynomial. In other words, if we let $\mathfrak{B}$ be the
solutions of~\eqref{equ:lambda_L} and
\begin{equation*}
  \mathfrak{S} = \left\{
    \prod_{i \in W} (1- A^i X) \in M_{\ell}(\F_{q^s}) :
      W \subset \{ 1,\ldots,m \} \text{ and }
      \#W \leq \halfdis
  \right\}
\end{equation*}
be the set of all possible locator polynomials corresponding
to errors of weight at most $\halfdis$,
we are interested in the elements of $\mathfrak{B} \cap \mathfrak{S}$.

\begin{proposition}
  There exists one and only one solution of
  equation~\eqref{equ:lambda_L} in $\mathfrak{S}$.
\end{proposition}

\begin{proof}
  Equation~\ref{equ:KE} ensures that there exists at least one element
  in $\mathfrak{B} \cap \mathfrak{S}$. If there were more than one
  solution in $\mathfrak{S}$ there would exist more than one codeword
  in a Hamming ball of radius $\halfdis$ which
  is absurd.
\end{proof}

The solving of~\eqref{equ:lambda_L} remains difficult. One
needs an exponential (in $\ell\delta$) number of arithmetic operations
in $\F_{q^s}$ to find the element of $\mathfrak{B} \cap \mathfrak{S}$.
For small values of $q$, $\ell$ and $\delta$ the solution can be found
by exhaustive search on the solutions of~\eqref{equ:lambda_L}.

\subsubsection{Unambiguous decoding scheme}

In this subsection, we prove that, as in the BCH case, the roots of
the locator polynomial (in $\F_{q^s}[A]$) give precious
information about the location of errors. The factorization of
polynomials of $M_\ell(\F_{q^s})[X]$ is not unique, all the
roots of the locator polynomial do not indicate an error position.

\begin{proposition}
\label{prop:equivRacineLoca}
  Let $e \in \F_q^{m\ell}$ be an error vector such that
  $w(e) \leq \halfdis$ and $\Lambda(X)$ be
  the locator polynomial associated to $e$. We have 
  \begin{equation*}
    e_i \neq 0 \Longleftrightarrow \Lambda(A^{-i}) = 0.
  \end{equation*}
\end{proposition}

\begin{proof}
  By definition, we have $\Lambda(A^{-i}) = 0$ if $e_i \neq 0$.
  Conversely, if $e_i = 0$ then $A^j A^{-i} \neq I_{\ell}$ for
  $j \in \supp(e)$. Thus $1 - A^j A^{-i}$ is a unit in
  $\F_{q^s}[A]$ by definition of $A$. Therefore
  $\Lambda(A^{-i}) \neq 0$.
\end{proof}
  
These roots can be found by an exhaustive search on the powers of $A$
in at most $m$ attempts. At this step the support of the error vector
$e$ is known. The last step to complete the decoding is to find the
value of the error.

\begin{proposition}
\label{prop:errorEvaluation}
  Let $e \in \F_q^{m\ell}$ be an error such that
  $w(e) \leq \halfdis$, \mbox{$W=\supp(e)$}, $\Lambda(X)$ be the
  locator and $L(X)$ be the evaluator polynomials associated to
  $e$. If $A^{-i}$ is a root of $\Lambda(X)$ for $ i \in W$,
  then 
  \begin{equation*}
    e_i =  \prod_{j \in W \setminus \{i\}}
                (A^{i} - A^{j})^{-1} L(A^{-i})
  \end{equation*}
  where $L(A^j)$ denotes $\sum (A^j)^i L_i$.
\end{proposition}

\begin{proof}
  Let $i_0 \in W$. We have
  \begin{align*}
    L(A^{-i_0}) &= \sum_{i = 1}^{w} \prod_{j \neq i}^{w}
                   A_i (1 - A^{-i_0} A_j) y_i\\
                &= \prod_{j \in W \setminus \{i_0\}}
                   A^{i_0} (1 - A^{-i_0} A^j) e_{i_0}\\
                &= \prod_{j \in W \setminus \{i_0\}}
                   (A^{i_0} - A^j) e_{i_0}.
  \end{align*}
  By definition of $A$, $A^{i_0} - A^{j}$ is invertible for all
  $j \in W$ hence the result.
\end{proof}

\begin{algorithm}
\label{al:DecodingQCBCH}
\caption{Decoding algorithm for quasi-BCH codes}
\begin{algorithmic}
  \REQUIRE{The received word $y = c + e$ where
           $c \in \C$ and $w(e) \leq \halfdis$.}
  \ENSURE{The codeword $c$, if it exists such that
          $d(y,c) \leq \halfdis$.}
  \STATE $S_{\delta}(X) \gets$ Syndrome of $y$.
  \STATE Compute $\Lambda(X)$ and $L(X)$
         (Subsection~\ref{sss:solving_ke}).
  \STATE $\mathfrak{R} \gets$ roots of $\Lambda(X)$ in $\F_{q^s}[A]$.
  \STATE $W \gets \{i | A^{-i} \in \mathfrak{R} \}$.
  \STATE $\zeta \gets (0,\ldots,0)$.
  \FOR{$i \in W$}
    \STATE $\zeta_i = \prod_{j \in W \setminus \{i\}}
                      (A^{i} - A^{j})^{-1} L(A^{-i})$.
  \ENDFOR
  \RETURN $y - \zeta$.
\end{algorithmic}
\end{algorithm}

\section{Evaluation codes}
\label{Sec:Eval}
\subsection{Definition and parameters}

In this subsection we generalize evaluation codes. For any ring $R$
and any positive integer $k$, we denote by $R[X]_{<k}$ the left
$R$-module of all polynomials of $R[X]$ of degree at most $k - 1$.

\begin{proposition}
\label{FqAField}
  Let $q$ be a prime power and $\ell,m$ be positive integers such
  that $m = q^{\ell} - 1$. Let $A \in M_{\ell}(\F_q)$ be a primitive
  $m$-th root of unity. Then $\F_q[A]$ and $\F_{q^{\ell}}$ are
  isomorphic as rings.
\end{proposition}

\begin{proof}
  Let $\mu(X)$ be the minimal polynomial of $A$ of degree at
  most~$\ell$. We have $\mu | X^m - 1$, thus the roots of $\mu$
  are all distinct. By Definition~\ref{defi:racprim}-~(3), the
  roots of $\mu$ lie in $\F_{q^{\ell}}$ and not in any subfield.
  Therefore $\mu$ is irreducible.
\end{proof} 

\begin{definition}[Quasi-cyclic evaluation codes]
\label{defi:evalcode}
  Let $\ell$ be a positive integer and $q$ be a prime power. Let
  $m = q^{\ell} - 1$ and $k \leq m$. Let $A \in M_\ell(\F_q)$ a
  primitive $m$-th root of unity.
  Let $\pi$ be a $\F_q$-linear map from $\F_q[A]$
  into $\F_q^{\ell}$. We denote by $C_{A,k,\pi}$ the image of:
  \begin{equation*}
    \begin{array}{ccccc}
      (\F_q[A])[X]_{<k} & \stackrel{\ev_A}{\longrightarrow} &
        (\F_q[A])^m & \stackrel{\pi^{\times m}}{\longrightarrow} &
          (\F_q^\ell)^m \\
      P(X) & \longmapsto & \left(P(A^0),\ldots,P(A^{m-1})\right) &
        \longmapsto &
          \left(\pi(P(A^0)),\ldots,\pi(P(A^{m-1}))\right).
    \end{array}
  \end{equation*}
\end{definition}

\begin{proposition}
  Taking the notation of Definition~\ref{defi:evalcode},
  $C_{A,k,\pi}$ is a $\ell$-quasi cyclic code over $\F_q$ of
  length $m\ell$ and of dimension over $\F_q$ at least 
  $k \ell - \dim_{\F_q}(\ker \pi^{\times m})$.
\end{proposition}

\begin{proof}
  By Proposition~\ref{FqAField} the statement about the dimension of
  $C_{A,k,\pi}$ is obvious. Let
  \begin{equation*}
    P(X) = \sum_{i = 0}^{k - 1} \sum_{j = 0}^{m - 1} P_{ij} A^j X^i
    \in \F_q[A][X]_{<k}
  \end{equation*}
  with $P_{ij} \in \F_q$.
  Then
  \begin{equation*}
    Q(X) = \sum_{i = 0}^{k - 1} \sum_{j = 0}^{m - 1} P_{ij}
           A^{j + i} X^i
    \in \F_q[A][X]_{<k}
  \end{equation*}
  is such that $Q(A^i) = P(A^{i + 1})$ for all $i \in \Z$
  and $C_{A,k,\pi}$ is $\ell$-quasi cyclic.
\end{proof}

\subsection{New good codes}

\begin{proposition}
  Using the notation of Definition~\ref{defi:evalcode},
  if $\pi$ is such that for $B=(b_{ij}) \in \F_q[A]$
  \begin{itemize}
    \item $\pi(B)=(b_{i1},\ldots,b_{i\ell})$ for some $i$,
    \item or $\pi(B)=(b_{1j},\ldots,b_{\ell j})$ for some $j$,
  \end{itemize}
  then $\dim C_{A,k,\pi} \geq k \ell$ and $C_{A,k,\pi}$ has minimum
  distance $d \geq m - k + 1$.
\end{proposition}

\begin{proof}
  In both cases, it suffices to notice that $\pi^{\times m}$ is
  injective. If $\pi^{\times m}(B_1,\ldots,B_m) = 0$ then
  $\det B_i = 0$ for $i = 1,\ldots,m$. As $\F_q[A]$ is a field we must
  have $B_i = 0$ for $i = 1,\ldots,m$. In fact under the assumptions
  of the proposition $\pi^{\times m}$ is an isomorphism since
  $\#((F_q[A])^m) = q^{m \ell} = \#((F_q^{\ell})^m)$.
\end{proof}

\begin{remark}
\label{rem:new_codes}
  \begin{enumerate}
    All the computations of the examples below have been performed
    with the \textsc{magma} computer algebra system \cite{magma}.
    \item \label{rem:item:ze_one}
          For some particular choices of $\pi$, especially when we
          decrease the dimension~$k$, we observe that the minimum
          distance is multiplied by $\ell - 1$. For example, with
          \begin{equation*}
            A =
            \begin{pmatrix}
              0      & \omega   & 0 \\
              \omega & \omega^2 & \omega^2 \\
              1      & \omega^2 & 1
            \end{pmatrix} \in M_3(\F_4)
            \text{ with } \F_4 = \F_2[\omega],
          \end{equation*}
          $k = 4$ and $\pi((b_{ij}))=(b_{2,1},b_{1,2},b_{2,3})$,
          we find a $[189,11,125]_{\F_4}$-code. According to
          \cite{codetables}, the previous best known minimum
          distance was $121$.
    \item \label{rem:item:rs}
          As for Reed-Solomon codes, we can evaluate polynomials of
          $(\F_q[A])[X]_{<k}$ at less than $m = q^{\ell} - 1$ points.
          Following this approach, we find the following new good
          codes listed below together with the corresponding previous
          best known minimum distances:
          \begin{center}
            $[186, 11, 122]_{\F_4}$, 120; \\
            $[183, 11, 119]_{\F_4}$, 117; \\
            $[180, 11, 116]_{\F_4}$, 114; \\
            $[177, 11, 113]_{\F_4}$, 112. \\
          \end{center}
    \item Markus Grassl applied different methods to construct new codes
          from our $[189, 11, 125]_{\F_4}$ code
          (item~\ref{rem:item:ze_one} of Remark~\ref{rem:new_codes}).
          For example, he used a puncturing method \cite{GrasWhi2004}.
          Some of the codes he obtained have the same parameters as the
          codes listed in item~\ref{rem:item:rs} of
          Remark~\ref{rem:new_codes}. He found 
          $[186, 11, 122]_{\F_4}$, $[183, 11, 119]_{\F_4}$ and
          $[180, 11, 116]_{\F_4}$ codes. He also found a
          $[177, 11, 114]_{\F_4}$ code while the best known
          minimum distance was~$112$.
          The 49 new codes found with the help of Markus Grassl are
          listed in Table~\ref{tab:BestCodes}. All the methods used
          for the construction of these codes are detailed
          in~\cite{codetables}.
          \begin{table}
            \begin{tabular}{|ccccc|}
              \hline
              \multicolumn{5}{|c|}{New codes over $\F_{4}$} \\
              \hline
              \hline
              $[171, 11, 109]_4$ & $[172, 11, 110]_4$ &
              $[173, 11, 110]_4$ &
              $[174, 11, 111]_4$ & $[175, 11, 112]_4$ \\
      
              $[176,11,113]_4$ & $[177,11,114]_4$ &
              $[178,11,115]_4$ &
              $[179,11,115]_4$ & $[180,11,116]_4$ \\
          
              $[181, 11, 117]_4$ & $[182, 11, 118]_4$ &
              $[183, 11, 119]_4$ &
              $[184, 10, 121]_4$ & $[184, 11, 120]_4$ \\
          
              $[185,10,122]_4$ & $[185,11,121]_4$ &
              $[186,10,123]_4$ &
              $[186,11,122]_4$ & $[187,10,124]_4$ \\
          
              $[187,11,123]_4$ & $[188,10,125]_4$ &
              $[188,11,124]_4$ &
              $[189,10,126]_4$ & $[189,11,125]_4$ \\
              
              $[190,10,127]_4$ & $[190,11,126]_4$ &
              $[191,10,128]_4$ &
              $[191,11,127]_4$ & $[192,11,128]_4$ \\
          
              $[193,11,128]_4$ & $[194,11,128]_4$ &
              $[195,11,128]_4$ &
              $[196,11,129]_4$ & $[197,11,130]_4$ \\
          
              $[198,11,130]_4$ & $[199,11,131]_4$ &
              $[200,11,132]_4$ &
              $[201,10,133]_4$ & $[201,11,132]_4$ \\
          
              $[202,10,134]_4$ & $[202,11,132]_4$ &
              $[203,10,135]_4$ &
              $[204,10,136]_4$ & $[204,11,133]_4$ \\
          
              $[205,11,134]_4$ & $[210,11,137]_4$ &
              $[213,11,139]_4$ & $[214,11,140]_4$ & \\ 
          \hline
        \end{tabular}
        \caption{\label{tab:BestCodes}
          49 new codes over $\F_4$ which have a larger minimum
          distance than the previously known ones.}
        \end{table}
  \end{enumerate}
\end{remark}

\begin{remark}
  We have proved in Proposition~\ref{FqAField} that $\F_q[A]$ is a
  field such that $[\F_q[A]:\F_q] = \ell$. Thus there is a
  $\F_q$-linear isomorphism from $\F_q[A]$ to $\F_q^{\ell}$.
  Consider the following one:
  \begin{equation*}
    \begin{array}{ccc}
      \F_q[A] & \stackrel{\psi}{\longrightarrow} & \F_q^{\ell} \\ 
      B=b_0 I_{\ell} + b_1 A + \cdots + b_{\ell - 1} A^{\ell - 1} &
          \longmapsto & (b_0,b_1,\ldots,b_{\ell - 1}).
    \end{array}
  \end{equation*}
  Then
  \begin{equation*}
    C_{A,k,\psi}=\psi^{\times m}(\ev_A(\F_q[A][X]_{<k}))
  \end{equation*}
  is still an $\ell$-quasi cyclic code of length $m\ell$ and of
  dimension $k\ell$. Let $\Pi \in M_{\ell}(\F_q)$ and let
  \begin{equation*}
    \begin{array}{rcl}
      \pi : \F_q^{\ell} & \rightarrow & \F_q^{\ell} \\
                      x & \mapsto     & x \Pi
    \end{array}
  \end{equation*}
  for a given $\Pi \in M_{\ell}(\F_q)$.
  Then
  \begin{equation*}
    C_{A,k,\psi,\pi} =
      \pi^{\times m}(\psi^{\times m}(\ev_A(\F_q[A][X]_{<k})))
  \end{equation*}
  is an $\ell$-quasi cyclic code of length $m\ell$ and dimension
  $\geq k \ell - \dim(\ker \pi)$.

  We notice that there exist matrices $\Pi$ for which the obtained
  minimum distance is always greater than $m - k + 1$. For instance,
  taking $\ell = 3$, $q = 4$ and the matrix
  \begin{equation*}
    \Pi =
    \begin{pmatrix}
      1        & \omega^2 & \omega \\
      \omega^2 & \omega   & 1 \\
      1        & 1        & 1
    \end{pmatrix},
   \end{equation*}
   give codes with minimum distance close to $2(m - k + 1)$.
\end{remark}

\section{Conclusion}
\label{Sec:Conlu}
In this paper we presented a generalization of results for cyclic
codes to quasi-cyclic codes. We proved that there is a natural
one-to-one correspondence between $\ell$-quasi-cyclic codes and
left ideals of $M_{\ell}(\F_q)[X] / (X^m - 1)$.
We then extended the construction of BCH and evaluation codes
to this context. This generalization allowed us to find a lot of
new codes with good parameters and, sometimes, beating previous
known minimum distances. A deeper study of decoding algorithms for
quasi-BCH need more work and remains an open problem.

\section*{Acknowledgments} 
We would like to thank the referees, whose suggestions have
permit to improve this article, in particular, for
Algorithm~\ref{al:basis-block-rank} and for the idea of using
the Morita equivalence to prove the one-to-one correspondence
between left ideals and quasi-cyclic codes.
We would like to thank Markus Grassl for his precious help for
finding new good codes.

\bibliographystyle{plain}
\bibliography{biblio2}

\end{document}